\documentclass[12pt, draftclsnofoot, onecolumn]{IEEEtran}
\usepackage{mathtools}
\usepackage{amsmath,color,amsthm}
\usepackage{breqn}
\usepackage{bm}
\usepackage{cite}
\usepackage{algorithm,algpseudocode}
\usepackage{caption}
\usepackage{lipsum}
\usepackage{cuted}
\usepackage{enumerate}
\usepackage{booktabs} 
\usepackage{graphicx}
\usepackage{float}
\usepackage{mathrsfs}

\usepackage{amsmath,amsthm,amssymb}
\newtheorem{proposition}{Proposition} 

\newtheorem{theorem}{Theorem}
\newtheorem{corollary}{Corollary}
\newtheorem{lemma}{Lemma}

\usepackage{wasysym} 

\newcommand\givenbase[1][]{\:#1\lvert\:}
\let\given\givenbase

\DeclarePairedDelimiterX\Basics[1](){\let\given\sgiven #1}

\date{\today}

\begin{document}
	\title{ Generalized Coordinated Multipoint (GCoMP)-Enabled NOMA: Outage, Capacity, and Power Allocation
    } 
\author{Yasser Al-Eryani, Ekram Hossain, and Dong In Kim\thanks{Y. Al-Eryani and E. Hossain are with the Department of Electrical and Computer Engineering at the University of Manitoba, Canada (Emails: aleryany@myumanitoba.ca, Ekram.Hossain@umanitoba.ca). D. I. Kim is with the School of Information and Communication
Engineering at the Sungkyunkwan University (SKKU), Korea (Email: dikim@skku.ac.kr). 
The work was supported by a Discovery Grant form the Natural Sciences and Engineering Research Council of Canada (NSERC) and in part by  the National Research Foundation of Korea (NRF) grant funded by the Korean government under grants 2017R1A2B2003953 and 2014R1A5A1011478.}}
 	\maketitle
	\begin{abstract}

A novel generalized coordinated multi-point transmission (GCoMP)-enabled non-orthogonal multiple access (NOMA) scheme is proposed.
In particular, the traditional joint transmission CoMP scheme is generalized to be applied for all user-equipments (UEs), i.e. both cell-centre and cell-edge users within the coverage area of cellular base stations (BSs). Furthermore, every BS applies NOMA for all UEs associated to it using the same frequency sub-band (i.e. all UEs associated to a BS forms a single NOMA cluster).
To evaluate the proposed scheme, we derive a closed-form expression for the probability of outage for a UE with different orders of BS cooperation.
Important insights on the proposed system are extracted by deriving an approximate (asymptotic) expressions for the probability of outage and outage capacity.
Furthermore, an optimal transmission power allocation scheme that jointly allocates transmission power fractions from all cooperating BSs to all connected UEs is developed and investigated for the proposed system.
Findings show that NOMA with a large number of UEs is feasible when the GCoMP technique is used over all UEs within the network coverage area. 
Also, the performance degradation caused by a large NOMA cluster size is significantly mitigated by increasing the number of cooperating BSs.
In addition, for given feasible system parameters and a given NOMA cluster, the lower the available power budget, the higher is the number of BSs that apply NOMA for their cluster members and the lower the number of BSs that use water-filling for power allocation.  
		
	\end{abstract}
	\section{Introduction}
To mitigate the problems of spectrum scarcity and interference in dense cellular networks, different base stations (BSs) can cooperate with each other either by coordinating their transmission beams (joint scheduling and beamforming) or by serving the same user equipment (UE) using the same time and frequency resources (joint transmission)  \cite{CoMP6146494}. The technique of BS coordination is referred to as coordinated multipoint (CoMP) transmission \cite{CoMP2011}. 
With CoMP transmission, either multiple BSs cooperate to serve a single cell-edge UE using the same spectrum resources or they coordinate their transmissions such that inter-cell interference (ICI) is minimized \cite{6571307}.   
One major drawback of Joint Transmission CoMP (JT-CoMP) scheme is the spectrum inefficiency due to the allocation of frequency resources of multiple BSs to a single UE. 
This has limited the application of CoMP for only UEs which do not have any dominant serving BS (i.e. cell-edge UEs). 
Nevertheless, it is possible to apply BS cooperation to both cell-edge and cell-centre UEs in order to boost up the overall transmission rate.

Recently, power-domain non-orthogonal multiple access (NOMA) has emerged as a promising technology to enhance spectrum efficiency of  both uplink and downlink cellular wireless networks \cite{NOMA2006,8085125,NOMA5G}. 
With NOMA, at a certain frequency sub-band, signals of multiple users are superimposed in the power domain such that the received signal for each UE has a distinct power level. 
At the NOMA receiver's end, successive interference cancellation (SIC) is used to cancel signal components with higher weights than the desired signal (starting from the signal with the highest weight) \cite{Verdu}. 
Theoretically speaking, NOMA enhances the spectral efficiency significantly compared to that of orthogonal multiple access schemes (OMA) at the expense of receiver complexity and processing delay \cite{NOMAPER2017,8316261}.
Generally, NOMA can be applied either in the uplink or downlink of a wireless system.
In this paper, we focus on downlink NOMA. However, the concepts and analysis can easily be applied to uplink NOMA.


The `marriage' between CoMP and NOMA introduces an attractive solution to compensate for the excessive spectrum usage in CoMP \cite{CoMPNOMA2014}. 
It also makes it possible to generalize the concept of CoMP of cooperation among BSs to serve all UEs within the network instead of just cell-edge UEs. 
Furthermore, the complexity level of NOMA receivers due to the SIC requirement can be reduced through enhancing the UE's signal-to-interference-plus-noise ratio (SINR), especially when a large NOMA-cluster size (e.g. more than two UEs in a cluster) is used.

CoMP-enabled NOMA has been investigated by researchers recently \cite{Fu2017NOMA,ChoiCoMPNOMA2014,YueTianCoMPNOMA2017}. 
In \cite{Fu2017NOMA}, a distributed power allocation scheme was investigated for cooperating BSs using NOMA. 
A coordinated superposition coding scheme for a two-BS downlink network was introduced in \cite{ChoiCoMPNOMA2014}. 
In \cite{ChoiCoMPNOMA2014}, it was shown that joint transmission CoMP  with two BSs allows NOMA to provide a common cell-edge UE with a reasonable transmission rate without sacrificing the rate of cell-centre UEs. Additionally, a multi-tier NOMA strategy in CoMP network was investigated in \cite{YueTianCoMPNOMA2017} where it was shown that performance enhancement can be achieved in multi-tier CoMP-enabled NOMA networks when a proper BS selection is performed.
Additionally, limited scenarios of CoMP-enabled NOMA were studied where groups of two UEs are capable of utilizing CoMP-NOMA techniques simultaneously \cite{Shipon2018,Shipon8352618}.
Interestingly, all of the previous works in the literature have restricted the application of CoMP to UEs that are physically located at cell-edge only.
However, cooperating BSs may be used to serve all UEs in their vicinity as long as the channel quality between these BSs and targeted UEs are acceptable. 
This generalized cooperation may be translated into an enhanced  SINR per UE at the expense of decreased spectral efficiency.

\subsection{Motivations and Contributions}
Even though the concept of CoMP has been very  promising, its adoption in cooperative wireless networks gained momentum due to significant enhancements in network infrastructure (e.g. fast backhauling based on optical fiber and free-space optical links). Future wireless networks (e.g. beyond 5G/6G systems) are expected to support virtual massive multiple-input multiple-output (massive MIMO)  communication along with cloud radio access network (C-RAN) technologies \cite{CRAN}. These should enable very fast and reliable communication among different wireless BSs and allow an efficient centralized/semi-centralized baseband processing for all BSs with minimum latency. Additionally, with the proliferation of internet of things (IoT) and massive machine-type communication (mMTC) services, every wireless  device will be connected to one or more wireless access networks to be served by multiple BSs, which in turn will be connected to a general cloud network to access cloud-based services (e.g. edge-computing and caching services). 
Such a massive connectivity will require more efficient utilization of the radio spectrum to increase the number of simultaneously served UEs per frequency band with minimized complexity and interference. In this context, adoption of NOMA in future networks is motivated by its potentials to enhance the spectrum efficiency significantly along with the rapid enhancements in  antenna technologies, sophisticated signal processing and SIC algorithms. 
Motivated by the potentials of both CoMP and NOMA technologies, we propose, design and analyze a novel generalized CoMP (GCoMP)-enabled NOMA scheme that allows cooperation among distributed BSs to serve all UEs within the network coverage area\footnote{The term `GCoMP' is used to differentiate the proposed cooperative scheme from the conventional CoMP systems in which nearby BSs cooperate to serve only the cell-edge UEs.}. 
The proposed scheme efficiently utilizes both joint transmission CoMP (JT-CoMP), which is usually used to improve the signal-to-interference-plus-noise ratio (SINR) of the cell-edge UEs and thus improve fairness in the network at the cost of using more bandwidth for a UE, and NOMA, that can improve the spectral efficiency. Generalizing JT-CoMP to GCoMP improves the SINR of all users (i.e. cell-edge and cell-centre users) in the coverage area and when combined with NOMA, it has the potential to significantly improve the overall spectral efficiency  performance.
Additionally, it enables the deployment of large-scale in-band NOMA clusters with a tolerable complexity level for SIC and acceptable performance.
The main contributions of this work can be summarized as follows: 
	   \begin{enumerate}[\Huge .]
	    \item We propose a novel channel access method referred to as the GCoMP-enabled NOMA and the corresponding $n{\text{-th}}$-order clustering scheme GCoMP-NOMA, where $n$ is the number of cooperating BSs per UE. 
		\item To evaluate the proposed GCoMP-NOMA scheme, utilizing the concepts of order statistics, we derive a closed-form expression for the downlink outage probability for a UE assuming independent but non-identically distributed (i.n.d) channel gains.
		\item In order to extract important insights on the performance of the proposed scheme, we derive simplified approximate expressions for the asymptotic outage probability and the asymptotic outage capacity per UE with full-order clustering.
		\item We design a low-complexity clustering protocol for the GCoMP-enabled NOMA system and develop an optimal transmission power allocation scheme per NOMA cluster.
	\end{enumerate}
The rest of this paper is organized as follows. 
Section II presents the general system model. 
The $n{\text{-th}}$ order clustering protocol for the GCoMP-enabled NOMA is presented in Section III. Section IV presents the outage and the outage capacity performance of the proposed system. The low-complexity clustering scheme and the corresponding power allocation method for the proposed system are presented in Section V.  Numerical results are presented in Section VI, followed by conclusions in Section VII.


\section{System Model}\label{s_sys}
Consider a downlink wireless network with $K$ BSs and $M$ UEs that are located at fixed locations within a certain network's physical area as shown in Fig. \ref{SMODEL}\footnote{Studying the random locations of BSs and UEs is out the scope of this work.}. 
All BSs are connected to each other through a fast backhaul link and have the ability of collaborate with each other at the baseband and radio levels.
This is similar to a C-RAN architecture in which distributed remote radio heads (RRHs) are connected to a single mega baseband unit (BBU) that has the ability of jointly processing signals from different RRHs by dealing with distributed RRHs as a virtual $K\times M$ MIMO system \cite{CRAN}. We also assume that perfect control signaling is possible among the cooperative BSs and the distributed UEs within their coverage area due to the existence of  fast backhauling links that connect all of the cooperative BSs together.

    	\begin{figure}[!htb]
		\centering
		\includegraphics[height=8cm, width=10cm]{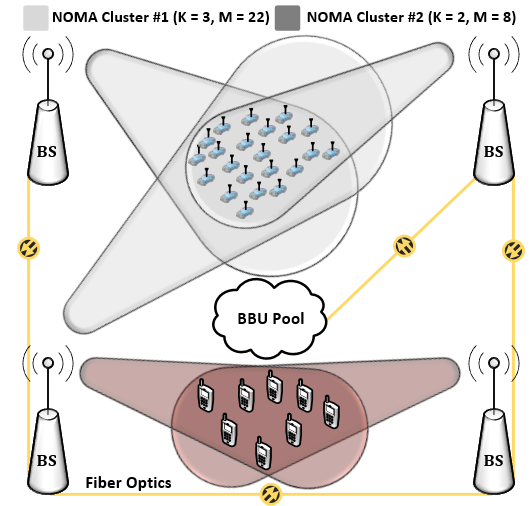}
		\caption{Example of the proposed network model. }\label{SMODEL}
	\end{figure}

For NOMA, we assume that every BS applies superposition coding in the power domain at every transmission sub-band (every sub-band has it's own power budget). 
Additionally, full channel state information (CSI) is assumed to be available at all BSs while CSI of every NOMA cluster is assumed to be available at all UEs of that cluster. A NOMA cluster refers to the set of UEs served by the same group of cooperating BSs using the same sub-band. 
The downlink channel gain between the $k{\text{-th}}$ BS and the $m{\text{-th}}$ UE is denoted by $h_{m,k}$. 
We assume that amplitudes of $h_{m,k}$ are independent but non-identically distributed (i.n.d) and follow complex Gaussian distribution, i.e. $|h_{m,k}|\thicksim  \mathcal {CG} (0, \sigma_m)$, where $\sigma_{m,k}$ is the standard deviation of $|h_{m,k}|$. 
Thus, the power gain $|h_{m,k}|^2$ follows an exponential distribution, i.e. $|h_{m,k}|^2 \thicksim  \text{Exp} (1/2\sigma_{m,k}^2)$.
At every BS, perfect SIC operation is assumed such that interference to every UE that is caused by UEs within the same NOMA cluster with lower channel gains are perfectly filtered out.   

Finally, when compared with a conventional CoMP system, we assume that the frequency reuse factor is always greater than one (e.g. 3, 7). This assumption is justified by the fact that such networks would be practical for crowded urban environments in which distributed BSs are close to each other which will cause severe interference to CoMP cell-centre UEs.

\section{Generalized CoMP-Enabled NOMA Scheme}
In this section, we present a clustering scheme for the GCoMP-enabled NOMA system which is referred to as the $n{\text{-th}}$ order clustering scheme.
\subsection{Proposal of an $n{\text{-th}}$ Order Clustering Scheme}
This is a variable order clustering scheme that enables multiple UEs to utilize a single NOMA sub-band with the help of joint cooperation among several connected BSs within their vicinity.
First, let $\Psi_{n,k}$ denote the set of indices of UEs connected to the $k{\text{-th}}$ BS throughout the same sub-band, where $n=1, \dots, K$ and $k=1, \dots, K$.
The clustering order $n$ denotes the number of BSs connected to every UE at a certain time instant. 
For example, if a $(2{\text{-nd}})$-order clustering is used, every UE within the network will be connected with two serving BSs. 
For the moment, given a certain channel gain matrix of the network $\mathbf{H} \in \mathbb C^{M\times K}$, the task is to find the set of UEs connected to the $k$-th BS under the $n$-th order clustering, $\Psi_{n,k}$.
Here, $\mathbf{H}$ is a matrix containing channel gains from all BSs to all UEs within the coverage area of the network.
\textbf{Algorithm 1} illustrates the proposed method that is used to find UEs' clusters for every BS under $n{\text{-th}}$ order clustering. 
\begin{algorithm}[H] 
\caption{: $n{\text{-th}}$ Order clustering for GCoMP-enabled $M-$user NOMA.}
\label{alg:loop1}
\begin{algorithmic}[1]
\Require{$\mathbf{H}\in \mathbb{C}^{M\times K}, n, K, M$} 
  \State {$\Psi_{i,k}=\Phi, \forall i=1, \dots, n$ and $k=1, \dots, K$}
    \For{$i=1:n$}   
        \For{$m=1:M$} 
        \State {$k^*=\max_k(\mathbf{H}(m,:))$}
        \State {$\Psi_{i,k^*}=\Psi_{i,k^*}\bigcup \{m\}$ and $ \mathbf{H}(m,k^*)=-1$}
        \EndFor
    \EndFor
    \State{$\Psi_{n,k}=\bigcup_{i=1}^{n}\Psi_{i,k}, k=1, \dots, K$}
    \State \Return {$\Psi_{n,\{1, \dots, K\}}$}
\end{algorithmic}
\end{algorithm}

Using the proposed clustering scheme, cooperation among distributed BSs is exploited for all UEs (both cell-edge and cell-centre UEs) within a certain network coverage area.
Additionally, for a certain sub-band, since every UE is a member of $n$ clusters, NOMA transmission power coefficients for every cluster must be allocated properly such that the desired received signal at a certain UE can be extracted through multilevel SIC operations.
This can be achieved by combining NOMA transmission power allocation of all BSs that are sharing the same UEs into one problem that is solved simultaneously either in a centralized or a distributed manner.
\subsection{NOMA Transmission for the Proposed Scheme}
At every NOMA sub-band, the $k{\text{-th}}$ BS will transmit $x_k=\sum_{m_k=1}^{|\Psi_{n,k}|}\sqrt{a_{m_k,k}P_k}s_{m_k,k}$, where $s_{m_k,k}$ is the message for the $m{\text{-th}}$ UE ($m_k$ is the ordered index of the $m{\text{-th}}$ UE in $\Psi_{n,k}$ and $m$ is an arbitrary index for UEs using the same sub-band within a certain network coverage area) with $\mathbb{E}\left[|s_{m_k,k}|^2\right]=1$, $P_k$ is the overall transmission power budget at the $k{\text{-th}}$ BS (assigned for a certain sub-band) and $a_{m_k,k}$ is the NOMA power allocation coefficient from the $k{\text{-th}}$ BS to the $m{\text{-th}}$ UE with index $m_k$ in $\Psi_{n,k}$ such that $\sum_{m_k=1}^{|\Psi_{n,k}|}a_{m_k,k} \leq 1$. 
Generally, the values of the coefficients $a_{m_k,k}, ~\forall~m_k=1, \dots, |\Psi_{n,k}|$ and $k=1, \dots, K$ can be derived optimally by considering joint optimization problem for all clusters within the network that uses the same sub-band.  A practical technique for optimizing the transmission power coefficients for the proposed GCoMP-enabled NOMA system will be presented later in this paper. 

First, let us define the received signal at the $m{\text{-th}}$ UE that uses a certain sub-band within the network coverage area as
\begin{equation}
y_m=\sum_{\forall~i\in \mathbb{S}_{\text{c},m}}h_{m,i}x_i
+
\sum_{\forall~j\in \mathbb{S}_{\text{nc},m}}h_{m,j}x_j+n_m,\label{Eq.1}
\end{equation}
where $n_{m}$ is the additive white Gaussian noise (AWGN) at the input of the $m{\text{-th}}$ UE with power spectral density (PSD) $N_{m}$ and $\mathbb{S}_{\text{c},m}$($\mathbb{S}_{\text{nc},m}$) is the set contains the indices of BSs serving (not-serving) the $m{\text{-th}}$ UE. 
As can be noticed from (\ref{Eq.1}), within a certain cluster set $\Psi_{n,k}$, every UE is an element of a random set of clusters (beside $\Psi_{n,k}$). 
Therefore, every UE in $\Psi_{n,k}$ will have a different inter-cell interference component that must be taken into consideration when establishing power fractions for NOMA transmission.
Therefore, let us assume that for the $k{\text{-th}}$ BS with a cluster set $\Psi_{n,k}$, the set of power gains between the $k{\text{-th}}$ BS and all UEs in $\Psi_{n,k}$ are normalized by inter-cell interference caused to every UE within the cluster and then ordered in an ascending mode, i.e. $\frac{|h_{1_k,k}|^2}{I_{\text{ICI}}^1}\leq \dots \leq \frac{|h_{m_k,k}|^2}{I_{\text{ICI}}^m} \leq \dots \frac{|h_{{|\Psi_{n,k}|_k},k}|^2}{I_{\text{ICI}}^{|\Psi_{n,k}|}}$, where $I^m_{\text{ICI}}$ is the NOMA inter-cluster interference results from the utilization of the same sub-band by other BS and is defined as $I_{\text{ICI}}^m=\sum_{\forall~i\in \mathbb{S}_{\text{nc},m}}\Phi_i |h_{m,i}|^2$ with $\Phi_i=P_i\sum_{j_i=1}^{|\Psi_{n,i}|}a_{j_i,i}$\footnote{Remember that $\sum_{j_i=1}^{|\Psi_{n,i}|}a_{j_i,i}=1$, however, we keep it in the analysis to illustrate the operations in the proposed model.}. 
Accordingly, for successful SIC operation at every UE, the transmission power allocation coefficients must be allocated such that the overall received power of the desired signal has a distinct power level from other combinations of undesired inter-NOMA interference signals (INI) with a certain power gap that depends on the sensitivity of the SIC hardware unit.
Accordingly, the design of SIC unit at every UE, which will require a multi-level SIC operation, may be challenging\footnote{While the purpose of this part of the work is to give an insight on the performance of multiuser NOMA under the proposed GCoMP-enabled NOMA model, a practical low-complexity GCoMP-enabled NOMA model will be presented along with the power allocation method later in this paper.}.

We assume that the power gains between the $m{\text{-th}}$ UE and different serving BSs are ordered such that $|h_{m,1}|^2\geq \dots \geq |h_{m,K}|^2$.
Hence, the {\rm SINR} at the input of the $m{\text{-th}}$ UE is given by
\begin{equation}
\gamma_{m}=\frac{\sum_{i=1}^n\Lambda_i^m|h_{m,i}|^2}
{
\underbrace{ \sum_{j=1}^n\Delta_j^m|h_{m,j}|^2}_{I_{\text{INI}}}+
\underbrace{ \sum_{w=n+1}^K\Phi_w|h_{m,w}|^2}_{I_{\text{ICI}}}+
\underbrace{N_m}_{\text{AWGN}}
},\label{Eq2}
\end{equation}
where $\Lambda_i^m=a_{m_i,i}P_i$, $\Delta_j^m=P_j\left(\sum_{l_j=m^*+1}^{|\Psi_{n,j}|}a_{l_j,j}+\Theta\right)$,  $\Phi_w=P_w\sum_{l_w=1}^{|\Psi_{n,w}|}a_{l_w,w}$,  $m^*=\max_k \left(m_k\right)$, $I_{\text{INI}}$ denotes the unfiltered INI component within the same NOMA cluster and $\Theta$ is the residual INI interference caused by the multi-level SIC.
The level of SIC is the number of required demodulation-subtraction operations within the SIC unit of a certain UE. 
Generally, $\Theta$ lies between $\Theta=0$ (when a certain UE has an identical ordering rank over all associated clusters) and $\Theta = \sum_{j=1}^n\left(\sum_{l_j=m_j+1}^{m_j+n}a_{l_j,j} \right)P_j|h_{m,j}|^2$ (when a certain UE has different ordering over all associated clusters).
However, the value of $\Theta$ can be set to zero with a proper sub-optimal design of the proposed system as will be discussed in the subsequent section.

Under this communication setup, the $m{\text{-th}}$ UE will detect the first $m^*-1$ signals using a first round of SIC, and then store the signal components that contain it's desired signal for a next level of SIC procedure (multi-level SIC is required only when $\Theta>0$).   
The available SINR at the input of the $m{\text{-th}}$ UE when detecting the signal component of the ${\delta}{\text{-th}}$ UE ($1\leq\delta\leq m^*$) is given by
\begin{equation}
    \gamma_{\delta\xrightarrow{}m}=\frac{\sum_{l=1}^n\Lambda_l^{\delta} \frac{ |h_{m,l}|^2}{\sum_{w=n+1}\Phi_{w}|h_{m,w}|^2+N_m}
    }
{
\sum_{j=1}^n\Delta_j^{\delta}\frac{|h_{m,j}|^2}{\sum_{w=n+1}\Phi_{w}|h_{m,w}|^2+N_m}+1
}.\label{Eq4}
\end{equation}
The difference between (\ref{Eq2}) and (\ref{Eq4}) appears at $\Lambda_j^{\delta}=a_{\delta,j}P_j$ and $\Delta_j^{\delta}=P_j\sum_{i_j={\delta}+1}^{|\Psi_{n,j}|}a_{i_j,j}$ and depends on the number of signals to be decoded and subtracted from the overall received signal. 
However, it is clear that $\gamma_{\delta\xrightarrow{}m}$ is an increasing function of $\delta$. 

\subsection{Illustrative Example}
To illustrate the main idea behind the GCoMP-enabled NOMA, we present a simple example with three BSs $(K=3)$ and eight UEs $(M=8)$ that operate under $2$-nd order clustering (i.e. $n=2$). 
Let us first assume a single realization of the channel gain matrix $\mathbf{H}$ as 
\begin{dmath}
\mathbf{H}^T=\begin{bmatrix}
 & \text{UE}_1 & \text{UE}_2 & \text{UE}_3 & \text{UE}_4 & \text{UE}_5 & \text{UE}_6 & \text{UE}_7 & \text{UE}_8 
    \\
   \text{BS}_1 & 0.9 & 1 & 0.45 & 0.7 & 0.39 & 1.2 & 0.38 & 0.89 
    \\
    \text{BS}_2 &0.1 & 0.98 & 0.35 & 0.65 & 0.93 & 0.72 & 0.91 & 0.3
    \\
    \text{BS}_3 &0.43 & 0.78 & 0.21 & 0.19 & 0.95 & 0.31 & 0.99 & 0.56
\end{bmatrix}.
\end{dmath}
Applying \textbf{Algorithm \ref{alg:loop1}} to $\mathbf{H}$, the decreasingly ordered cluster elements of $\Psi_{2,1}$, $\Psi_{2,2}$ and $\Psi_{2,3}$ is given in Table \ref{Table1}. 
\begin{table}[h!]
    \centering
      \caption{Cluster elements under $2\text{-nd}$-order clustering}
\begin{tabular}{c||c c c c c c}
\toprule[\heavyrulewidth]\toprule[\heavyrulewidth]
$\Psi_{2,1}$ & $\text{UE}_1$  & $\text{UE}_2$  & $\text{UE}_3$ & $\text{UE}_4$ & $\text{UE}_6$ & $\text{UE}_8$  \\
\hline
$\Psi_{2,2}$ & $\text{UE}_2$  & $\text{UE}_3$  & $\text{UE}_4$ & $\text{UE}_5$ & $\text{UE}_6$ & $\text{UE}_7$  \\
\hline
$\Psi_{2,3}$ & $\text{UE}_5$ & $\text{UE}_7$ & $\text{UE}_1$ & $\text{UE}_8$ & $\text{NA}$ & 
$\text{NA}$ \\  
\toprule[\heavyrulewidth]\toprule[\heavyrulewidth]
\end{tabular}
    \label{Table1}
\end{table}
For example, the received signal at $\text{UE}_1$ is given by
\begin{dmath}
y_1=\left(\sqrt{a_{1,1}P_1}h_{1,1}+\sqrt{a_{1,3}P_3}h_{1,3} \right)s_1
+
\left( \sqrt{a_{8,1}P_1}h_{1,1}+\sqrt{a_{8,3}P_3}h_{1,3}\right)s_8
+
 \sqrt{a_{2,1}P_1}h_{1,1}s_2
 +
 \sqrt{a_{3,1}P_1}h_{1,1}s_3
+
\sqrt{a_{4,1}P_1}h_{1,1}s_4
+
\sqrt{a_{5,1}P_1}h_{1,1}s_5
+
\sqrt{a_{6,1}P_1}h_{1,1}s_6
+
\sqrt{a_{7,1}P_1}h_{1,1}s_7
+n_1.
\end{dmath}
To extract $s_1$ from $y_1$, a set of SIC operations has to be conducted such that the overall power of the desired signal $s_1$ is separated by a certain SIC receiver sensitivity power gap $(P_s)$ from other undesired signal components.
The number of signals to be decoded before extracting $s_1$ will depend on the set of constraints required to extract $s_2$ through $s_8$ at $\text{UE}_2$ through $\text{UE}_8$.
To further illustrate, let us consider the case of full-order clustering at which every UE will be connected with all serving BSs in an ascending mode.
Accordingly, for the given  $\mathbf{H}$, the cluster sets are given in Table \ref{Table2}.
\begin{table}[h!]
    \centering
      \caption{Cluster sets under full order clustering}
\begin{tabular}{c||c c c c c c c c}
\toprule[\heavyrulewidth]\toprule[\heavyrulewidth]
$\Psi_{2,1}$ & $\text{UE}_1$  & $\text{UE}_2$  & $\text{UE}_3$ & $\text{UE}_4$ & $\text{UE}_6$ & $\text{UE}_8$ & $\text{UE}_5$ & $\text{UE}_7$   \\
\hline
$\Psi_{2,2}$ & $\text{UE}_2$  & $\text{UE}_3$  & $\text{UE}_4$ & $\text{UE}_5$ & $\text{UE}_6$ & $\text{UE}_7$ & $\text{UE}_1$ & $\text{UE}_8$  \\
\hline
$\Psi_{2,3}$ & $\text{UE}_5$ & $\text{UE}_7$ & $\text{UE}_1$ & $\text{UE}_8$ & $\text{UE}_2$ & $\text{UE}_3$ & $\text{UE}_4$ & $\text{UE}_6$ \\  
\toprule[\heavyrulewidth]\toprule[\heavyrulewidth]
\end{tabular}
    \label{Table2}
\end{table}
Therefore, the received signal at $\text{UE}_1$ is given by
\begin{dmath}
y_1=
\left(\sqrt{a_{1,1}P_1}h_{1,1}+
\sqrt{a_{1,2}P_2}h_{1,2}+
\sqrt{a_{1,3}P_3}h_{1,3} \right)s_1
+
\left(\sqrt{a_{2,1}P_1}h_{1,1}+
\sqrt{a_{2,2}P_2}h_{1,2}+
\sqrt{a_{2,3}P_3}h_{1,3} \right)s_2
+
\left(\sqrt{a_{3,1}P_1}h_{1,1}+
\sqrt{a_{3,2}P_2}h_{1,2}+
\sqrt{a_{3,3}P_3}h_{1,3} \right)s_3
+
\left(\sqrt{a_{4,1}P_1}h_{1,1}+
\sqrt{a_{4,2}P_2}h_{1,2}+
\sqrt{a_{4,3}P_3}h_{1,3} \right)s_4
+
\left(\sqrt{a_{5,1}P_1}h_{1,1}+
\sqrt{a_{5,2}P_2}h_{1,2}+
\sqrt{a_{5,3}P_3}h_{1,3} \right)s_5
+
\left(\sqrt{a_{6,1}P_1}h_{1,1}+
\sqrt{a_{6,2}P_2}h_{1,2}+
\sqrt{a_{6,3}P_3}h_{1,3} \right)s_6
+
\left(\sqrt{a_{7,1}P_1}h_{1,1}+
\sqrt{a_{7,2}P_2}h_{1,2}+
\sqrt{a_{7,3}P_3}h_{1,3} \right)s_7
+
\left(\sqrt{a_{8,1}P_1}h_{1,1}+
\sqrt{a_{8,2}P_2}h_{1,2}+
\sqrt{a_{8,3}P_3}h_{1,3} \right)s_8
+
n_1.
\end{dmath}
In order to be able to extract $s_1$ from $y_1$, a set of SIC operations will be conducted. 
Specifically, for successful SIC operations, a subset of the following constraints must satisfied:
\begin{dmath}\label{SICConst2}
    |\mathcal{N}_i-\mathcal{N}_j|~\geq P_s,~\forall i\neq j, i=1, \dots, 8~\mbox{and}~j=1, \dots, 8,
\end{dmath}
where $\mathcal{N}_m=a_{m,1}P_1|h_{1,1}|^2+a_{m,2}P_2|h_{1,2}|^2+a_{m,3}P_3|h_{1,3}|^2$.
Note that a set of different constraints has to be satisfied at every UE and transmission power for NOMA must be allocated such that all constraints at all UEs are satisfied.
Furthermore, the number of SIC operations required at every UE will depend on it's overall power level (weight) compared to other distinct signals (the more the power the fewer will be the number of required SIC operations).

For simplicity of analysis, we adopt a constant power allocation scheme for NOMA under the assumption of the availability of perfect SIC for the proposed model.
In particular, the $k$-th BS applies a constant power allocation for its own cluster members served through NOMA based on the following relation:
\begin{dmath}
    a_{i_k,k}=
    \left\{
\begin{array}{cc}
\frac{1}{2^i}     &  i=1, \dots, |\Psi_{n,k}|-1,\\
    a_{|\Psi_{n,k}|-1,k} & i=|\Psi_{n,k}|.
\end{array}    
    \right.
\end{dmath}
Note that for this constant power allocation scheme, the condition $\sum_{i_k=1}^{|\Psi_{n,k}|}a_{i_k,k}=1$ holds. 
Besides, the notion of allocating higher power fractions to NOMA UEs with weaker channel gains is also satisfied.
Note that, the optimal power allocation scheme for a modified practical GCoMP-enabled NOMA system will be presented in Section \ref{PASection}.
\section{Outage and Capacity Performance of the Proposed Scheme}
In this section, a closed-form expression for the outage probability for a UE under the GCoMP-enabled NOMA scheme is derived. 
Generally, the $m{\text{-th}}$ UE will be in outage if it fails to successfully decode at least one of the $m$ higher weight signal components.
This can be mathematically expressed as\footnote{Here, we have assumed that the higher levels of SIC operations are conducted ideally such that $\Theta$ is set to zero. 
This is an assumption that complies with the practical design of the proposed system as will be shown later in this paper.}
\begin{dmath}
    \text{P}^{m}_{\text{out}}=1-\text{P}\left( \bigcap_{\delta=1}^{m} \text{E}_{{\delta}\xrightarrow{}m}^{\text c}  \right),\label{Eq.5}
\end{dmath}
where $\text E_{\delta\xrightarrow{}m}$ is  the event that the $m{\text{-th}}$ UE has failed to decode the $\delta{\text{-th}}$ signal component under a certain performance requirements and
$\text E^{\text c}_{\delta\xrightarrow{}m}$ is the complement of $\text E_{\delta\xrightarrow{}m}$.
This can be mathematically expressed as
\begin{dmath}
    \text{E}_{\delta\xrightarrow{}m}\overset{\Delta}=
    \left\{
\begin{array}{cc}
\text P \left( 
    \gamma_{\delta\xrightarrow{}m}\leq \gamma_{\text{th}}^{\delta}
    \right)     &  \delta=1,\\
    \text P \left( 
    \gamma_{\delta\xrightarrow{}m}\leq \gamma_{\text{th}}^{\delta}
    \given \gamma_{{\delta-1}\xrightarrow{}m}\leq \gamma_{\text{th}}^{\delta-1}
    \right) & \delta>1 ,
\end{array}    
    \right.
\end{dmath}
where $\gamma_{\text{th}}^{\delta}=2^{\bar R_{\delta}}-1$. The value $\bar R_{\delta}$ is the minimum transmission rate required by the ${\delta}{\text{-th}}$ UE (assuming that every UE has identical ordering over all connected clusters).
For simplicity of analysis, we will assume that all UEs have the same rate requirements, i.e. $\gamma_{\text{th}}^{\delta}=\gamma_{\text{th}}$.
Deriving a closed-form expression for (\ref{Eq.5}) is possible, however, the final expression is found to be complicated and difficult to be programmed. 
To simplify the analysis, we first calculate the average inter-cell interference at the $m{\text{-th}}$ UE and substitute it into (\ref{Eq4}).
\textbf{Theorem 1} defines the average value of $I_{\text{ICI}}^m$ of the proposed GCoMP-enabled NOMA scheme.
\begin{theorem}
   The average inter-cell interference at the $m{\text{-th}}$ UE under the GCoMP-enabled NOMA scheme is given by
\begin{dmath}
    \bar I^m_{\text{ICI}}=\sum_{i_{n+1}, \dots, i_{K}}^{\{1,2, \dots, K\}}\prod_{l=n+1}^K
    \frac{\lambda_{m,i_l}}{\sum_{d=n+1}^l\lambda_{m,i_d}+\sum_{q=1}^n\lambda_{m,i_q}}
    \left(\sum_{w=n+1}^{K}\frac{w}{\sum_{q=n+1}^w \lambda_{m,i_{q}}+\sum_{q=1}^n\lambda_{m,i_q}} \right),\label{IFainal}
\end{dmath}   
where $\lambda_{m,w}=1/2\Phi_w\sigma_{m,w}$ and $\{i_{n+1}, \dots, i_{K}\}$ are distinct indices that take values from $\{1, \dots, K\}$.\label{InterAve}
\end{theorem}
\begin{proof}
See \textbf{Appendix A}.
\end{proof}
\textbf{Remark:} \textbf{Theorem \ref{InterAve}}  considers only the ordering of BSs with respect to (w.r.t) any arbitrary UE at $\Psi_{n,\{1,\dots,K\}}$. 
This is true since $\bar I_{\text{ICI}}^m$ is constant for every UE even when the $m{\text{-th}}$ UE is decoding the $\delta{\text{-th}}$ UE's signal $ (\delta<m^*$).

Substituting $\bar I_{\text{ICI}}^m$ into (\ref{Eq4}), the outage probability of (\ref{Eq.5}) can be expressed as
\begin{equation}
    P_{\text{out}}^{m}=\text{P}\left(\sum_{i=1}^n z_{m,i}\leq \frac{\gamma_{\text{th}}P}{\max_{\delta}\left(\Lambda^{\delta}-\gamma_{\text{th}}\Delta^{\delta} \right)} \right),\label{out1}
\end{equation}
where $z_{m,i}\thicksim \text{Exp}(\alpha_{m,i})$, $\alpha_{m,i}=\left(\bar{I}_{\text{inter}}^m+N_m\right)/2P\sigma^2_{m,i}$ and we have assumed that all BSs are transmitting using the same maximum transmission power budget, i.e. $P_k=P, ~\forall~k=1, \dots, K$. 
Additionally, the maximum value of $\left(\Lambda^{\delta}-\gamma_{\text{th}}\Delta^{\delta} \right)$ changes at every time slot based on the instantaneous CSI and the  NOMA transmission power allocation method.
Accordingly, the probability of outage of the $m$-th UE under the GCoMP-enabled NOMA scheme is defined in \textbf{Theorem 2}.
\begin{theorem}
   The probability of outage of the $m{\text{-th}}$ UE under the $(n{\text{-th}})$-order clustering of the GCoMP-enabled NOMA scheme is given by
   
   \begin{dmath}
       \text{P}_{\text{out}}^{(n)}= \sum_{m=1}^{\Omega} \sum_{\textit{S}_m}\prod_{l=1}^m F^{(n)}_{\gamma_{\kappa_l}}\left(\gamma_{\text{th}}^'\right)\prod_{q=m+1}^{\Omega}\left[1-F^{(n)}_{\gamma_{\kappa_q}}\left(\gamma_{\text{th}}^'\right) \right],\label{OutFinal}
   \end{dmath}
   where the summation extends over all permutations $(\kappa_1, \dots, \kappa_{\Omega})$ of $1, \dots, \Omega$ for which $\kappa_1<\dots<\kappa_m$, $\kappa_{m+1}<\dots<\kappa_{\Omega}$ and $\Omega=\max_k |\Psi_{n,k}|$ such that $\{\Psi_{n,k}\big | m\in \Psi_{n,k}\}$. 
   The CDF $F^{(n)}_{\gamma_{\kappa}}(\gamma)$ is given by
   \begin{dmath}
    F^{(n)}_{{\gamma}_{\kappa}}(\gamma)=
    \sum_{i_{1}, \dots, i_{n}}^{\{1,2, \dots, K\}}J_1({\kappa},i)
    \left[
    \sum_{t_1=1}^n\frac {\eta_{t_1}^{\kappa}}{\rho_{t_1}^{\kappa}}\left(1-e^{-\rho^{\kappa}_{t_1}\gamma}\right)
    -
    \sum_{h_1=1}^{K-n}(-1)^{h_1}\sum^{\{n,\dots,K\}}_{j_1\leq \dots \leq j_{h_1}}
    \sum_{t_2=1}^n\frac
    {\eta_{t_2}^{\kappa}}{\rho_{t_2}^{\kappa}}\left(1-e^{-\rho_{t_2}^{\kappa}\gamma}\right) 
    \right],\label{PDFF}
\end{dmath}
where $J_1(\kappa,i)=\left( \prod_{q=1}^n \frac{\alpha_{\kappa,i_q}}{q}\right)$, $\gamma_{\text{th}}^'=\frac{ \gamma_{\text{th}}P}{\max_{\delta}\left(\Lambda^{\delta}-\gamma_{\text{th}}\Delta^{\delta} \right)}$ and $\eta_{t_1}^{\kappa}$, $\rho_{t_1}^{\kappa}$, $\eta_{t_2}^{\kappa}$ and $\rho_{t_2}^{\kappa}$ are defined in Appendix B.\label{Theorem2}
\end{theorem}
\begin{proof}
See \textbf{Appendix B}.
\end{proof}
\textbf{Remark:}
The effect of cooperation among BSs within GCoMP-enabled NOMA scheme appears in terms of the increased SINR and decreased $I_{\text{ICI}}^m$ per UE.
However, a better performance enhancement can be done by optimizing NOMA coefficients for all BSs simultaneously as one matrix, as will be discussed in a subsequent section.

One particular case of significant importance is that when all BSs within a certain geographical area cooperate to serve a set of UEs in their vicinity using the same sub-band.
This can be considered as a full-order clustering of the proposed scheme, i.e. $n=K$. 
\textbf{Corollary 1} presents a simpler expression of this particular scenario under the assumption of independent and identically distributed (i.i.d) channel gains. 
\begin{corollary}
Under full-order clustering ($n = K$) with i.i.d channel gains, the probability of outage of the GCoMP-enabled NOMA scheme given in \textbf{Theorem \ref{Theorem2}} reduces to
\begin{dmath}    P_{\text{out}}^{(K)}=\sum_{m=1}^{M} {M \choose m}\gamma(K,\alpha \gamma_{\text{th}^'})^m\left(
    1-\gamma(K,\alpha \gamma_{\text{th}}^')
    \right)^{Q},\label{POUTFULL}
\end{dmath}
where  $\alpha=N/2P\sigma^2$~$(\bar{I}_{\text{ICI}}=0)$, $Q=M-m$, $\gamma(x,y)$ is the normalized lower incomplete gamma function \cite[Eq. 6.5.2]{1965} and $\Gamma(x)$ is the gamma function.
\end{corollary}
\begin{proof}
This Corollary can easily be proven by repeating the same procedure of \textbf{Theorem \ref{Theorem2}} under the given assumptions.
\end{proof}

Furthermore, to get more insights about the performance of the proposed system in terms of the diversity order and coding gain, \textbf{Corollary \ref{Corollary2.2}} presents an approximate expression of the probability of outage of the $m$-th UE under full-order clustering at the high SINR regime.  
\begin{corollary}
Under full-order clustering with i.i.d channel gains and high SINR regime  $(\bar \gamma \xrightarrow{}\infty)$, the probability of outage of the GCoMP-enabled NOMA scheme can be expressed as
\begin{equation}
P_{\text{out}}^{(K)}\approx \left(\left(\frac{K\Gamma(K)}{M(\gamma_{\text{th}}^')^K}\right)^{1/K} \bar \gamma \right)^{-K},
\end{equation}
where $\bar \gamma=2P\sigma^2/N$.\label{Corollary2.2}
\end{corollary}
\begin{proof}
This can be proven by utilizing the series representation of $\gamma\left(K,\alpha \gamma_{\text{th}}^' \right)$ in \cite[Eq. 6.5.29]{1965}, substituting in (\ref{POUTFULL}), and then taking the first term (the dominant term). 
\end{proof}
\textbf{Remarks:} 
\begin{itemize}
    \item At high SINR,  $P_{\text{out}}^{(K)}\approx(G_c \bar{\gamma})^{-G_d}$, where $G_c$ and $G_d$ are  the coding gain and  diversity order, respectively  \cite{Giannakis2003tcom}.
    Accordingly, under full-order clustering we have $G_c=\left({K\Gamma(K)}/{M(\gamma_{\text{th}}^')^K}\right)^{1/K}$ and $G_d=K$.
    \item It can be noticed from \textbf{Corollary \ref{Corollary2.2}} that, besides $\gamma_{\text{th}}^'$, the outage performance of the system is significantly affected by the number of NOMA UEs used per cluster ($M$) in terms of the decreased coding gain.
    This negative impact caused by increasing the NOMA cluster size can be significantly annihilated by increasing K (when $K\xrightarrow{}\infty$, $G_c$ approaches 1).
\end{itemize}

\textbf{Outage Capacity:} We evaluate the achievable transmission rate per UE of the proposed GCoMP-enabled NOMA. 
However, a closed-form expression of the ergodic capacity of the proposed system is found to be very complicated and does not carry any significant insights.
Nevertheless, we evaluate the so called $\epsilon$-\textit{outage capacity} under the high SINR regime, where $\epsilon$ is the maximum allowable outage to achieve a capacity of $C_{\epsilon}$ \cite{goldsmith_2005}. 

\begin{proposition}
Under full-order clustering with i.i.d channel gains and high SINR regime $(\bar \gamma \xrightarrow{}\infty)$, the $\epsilon$-\textit{outage capacity} for a UE under GCoMP-enabled NOMA is approximated by
\begin{equation}
    C_{\epsilon}\approx \log_2\left(1+\sqrt[K]{\frac{\epsilon K \Gamma(K)}{M}}\bar \gamma \right).\label{Capacity}
\end{equation}
\end{proposition}
\begin{proof}
This can be directly proven by using the definition of $\epsilon$-\textit{outage capacity} in \cite{goldsmith_2005} and utilizing the approximate outage expression from \textbf{Corollary 2}.
\end{proof}

\section{Power Allocation Scheme for GCoMP-Enabled NOMA} \label{PASection}
In this section, a practical clustering protocol for GCoMP-enabled NOMA scheme is presented and an optimal transmission power allocation model for the proposed system is developed.
\subsection{A Low-Complexity Full-Order Clustering Scheme}
The exact model discussed in Section II is considered as the optimal scheme for GCoMP-enabled NOMA. 
However, this model assumes that any UE may have a different ordering w.r.t different serving BSs.
This will result in a relatively high complexity at a NOMA receiver since, in the worst case scenario (i.e. different ordering for every UE at all BSs), the set of constraints for successful SIC operation will increase significantly which will be reflected negatively in the SIC unit design.
Another issue of the optimal GCoMP-enabled NOMA is that it requires a complicated scheduling algorithm that first assigns the set of signals to be decoded at every power level and then allocates power fractions for these signals. 
This should be conducted considering the received signals of all UEs.

To solve this problem, we propose that every UE has the same ordering over all related clusters.
This can be achieved by defining a global channel-quality-based metric for every UE that takes all links between every UE and all connected BSs into consideration.
For simplicity of the following analysis, we will focus on the model of full-order clustering (i.e. $n=K$) in which all BSs cooperate to serve a set of $M$ UEs simultaneously.
\textbf{Algorithm \ref{Alg.2}} shows the proposed sub-optimal clustering protocol that finds a single cluster which contains the $M$ UEs served by $K$ BSs.  
\begin{algorithm}[H] 
\caption{: Sub-optimal $K\text{-th}$ order clustering for GCoMP-enabled NOMA.}
\label{Alg.2}
\begin{algorithmic}[1]
\Require{$\mathbf{H}\in \mathbb{C}^{M\times K}, K, M$} 
  \State 
  $h_{(K)}(m)=\sum_{k=1}^K|H(m,k)|^2, \forall~ m=1, \dots, M$
  \State
  $\Psi_K=\Phi$
    \For{$i=1:M$}   
        \State $m^*=\min_m \left(h_K \right)$
        \State {$\Psi_{K}=\Psi_{K}\bigcup \{m^*\}~\mbox{and}~h_K(m^*)=\infty$}
    \EndFor
    \State \Return {$\Psi_{n,\{1, \dots, K\}}$}
\end{algorithmic}
\end{algorithm}
The main idea of this method is to produce a `global' cluster vector that contains the ordered indices of the entire set of $M$ UEs.
Specifically, the norm of the  gain vector of any arbitrary UE and all connected BSs is utilized as the global ordering metric to find the cluster members.

After finding $\Psi_K$, the goal now is to formulate and solve the optimization problem to determine the transmission power coefficients ($a_{m,k}, \forall~m=1, \dots M~\mbox{and}~k=1, \dots K$) for the transmitted power from all cooperating BSs to the set of UEs in $\Psi_K$.
Based on the proposed system model, the norm metric of the UEs will be ordered such that ${||\mathbf{H}(1,[1,\dots,K])||_2}\leq {||\mathbf{H}(2,[1,\dots,K])||_2}\leq \dots \leq {||\mathbf{H}(M,[1,\dots,K])||_2}$. Note that when lower order clustering is used (i.e. $n<K$), the link quality of the $m$-th UE should be divided by the $\bar I_{\text{ICI}}^m$. Accordingly, the optimization problem can be formulated as
\begin{equation}
\begin{aligned}
& ~\textbf{J}:~ \underset{a_{m,k}}{\text{max}}
& \text{\hspace{-75mm}} \sum_{m=1}^{M}\log_2\left(1+\frac{\sum_{k=1}^K a_{m,k}\gamma_{m,k}}{\sum_{k=1}^{K}\left(\sum_{j=m+1}^{M}a_{j,k}\right)\gamma_{m,k}+1} \right)\\
& ~\text{Subject to:} \\
&  ~ \textbf{C}_1:\log_2\left(1+\frac{\sum_{k=1}^K a_{m,k}\gamma_{m,k}}{\sum_{k=1}^{K}\left(\sum_{j=m+1}^{M}a_{j,k}\right)\gamma_{m,k}+1} \right)\geq R_m, \quad \forall m\\
& ~  \textbf{C}_2: \sum_{k=1}^{K} \left(a_{\delta_l,k}-\sum_{i=\delta_l+1}^{l}a_{i,k} \right)\gamma_{l,k}\geq P_s,
\\
&  ~  \textbf{C}_3:\sum_{m=1}^M a_{m,k}\leq 1,~\forall~ k=1, \dots, K,
\\
&~~  \forall~\delta_l= 1, \dots, l-1\;\mbox{and}\; l=2, \dots, M,\\
\end{aligned}\label{OptimizationProb}
\end{equation}
where $\gamma_{m,k}=P_K|h_{m,k}|^2/N_m$,  $R_m$ is the minimum required normalized transmission rate for a UE and is represented by the condition $\textbf{C}_1$, $\textbf{C}_2$ refers to the set of $\sum_{l=2}^M(l-1)=\frac{M(M-1)}{2}$ conditions required for successful SIC operation with receiver sensitivity of $P_s$, and $\textbf{C}_3$ represents the set of $M$ conditions related to the maximum power budget per BS.
The optimization problem \textbf{J} in  (\ref{OptimizationProb}) is a multiuser sum-rate maximization problem in an interference-limited environment. 
In general, these type of problems are non-convex due to the existence of dependent variables at the denominator of the SINR, which creates a sort of random convex-concave isolation of the objective function.
However, by the introduction of NOMA condition to these type of problems (e.g.  set of constraints given by $\textbf{C}_2$), the convexity status of the overall objective function changes.
\textbf{Lemma \ref{LemmaConcave}} states the convexity status of the problem in (\ref{OptimizationProb}). 
\begin{lemma}\label{LemmaConcave}
  Given the proposed GCoMP-enabled NOMA scheme, the optimization problem \textbf{J} formulated in (\ref{OptimizationProb}), which maximizes the normalized sum-rate of $M$ NOMA UEs per cluster, is a convex problem. 
\end{lemma}
\begin{proof}
See \textbf{Appendix C}.
\end{proof}
Due to the convexity of problem (\ref{OptimizationProb}), a closed-form expression for optimal transmission power fractions $\{a_{m,k}\}_{m=1, \dots, M}^{k=1, \dots, K}$ can be derived using the Lagrangian multipliers method as follows.
For the simplicity of analysis, we will illustrate the derivation of problem \textbf{J} at (\ref{OptimizationProb}) with relatively small system parameters ($M=3$ and $K=2$); however, the same procedure can be used to generalize the solution for any set of parameters. 
The Lagrangian function of problem \textbf{J} in (\ref{OptimizationProb}) can then be written as 
\begin{dmath}\label{Lagrangian}
    \mathcal{L}\left( \bm{a}, \bm{\eta}, \bm{\mu}, \bm{\tau}\right)=
    \sum_{m=1}^{3}\log_2\left(1+\frac{\sum_{k=1}^2 a_{m,k}\gamma_{m,k}}{\sum_{k=1}^{2}\left(\sum_{j=m+1}^{3}a_{j,k}\right)\gamma_{m,k}+1} \right)
    +
    \sum_{m=1}^3\eta_m
    \left[
    \sum_{k=1}^2
    \left(
    a_{m,k}-\gamma_{\text{th}}^m\sum_{j=m+1}^3 a_{j,k}
    \right)
    \gamma_{m,k}
    -\gamma_{\text{th}}^m
    \right]
    +
    \mu_1
    \left[
    P_s-
    \sum_{k=1}^2
    \left(
    a_{1,k}-a_{2,k}
    \right)\gamma_{2,k}
    \right]
    +
    \sum_{i=2}^3 \mu_i
    \left[
    P_s-
    \sum_{k=1}^2
    \left(
    a_{i-1,k}-\sum_{j=i}^3 a_{j,k}
    \right)\gamma_{3,k}
    \right]
    +
    \sum_{k=1}^2 \tau_k 
    \left[
    1-\sum_{m=1}^3 a_{m,k}
    \right],
\end{dmath}
where $\bm{a}=\left\{a_{i,j} \right\}_{i=1,2,3}^{j=1,2}$, and $\eta\geq 0$, $\mu\geq 0$ and $\tau\geq 0$ are the Lagrange multipliers corresponding to $\textbf{C}_1$, $\textbf{C}_2$, and $\textbf{C}_3$, respectively.
 Further discussion on the solution of Problem \textbf{J} in (\ref{OptimizationProb}) is given in \textbf{Appendix D}.




Finally, it is important to mention that for the proposed GCoMP-enabled NOMA scheme, in practical scenarios, the set of cooperating BSs may not be able to provide every UE within their cluster with its minimum rate requirement due to the power budget limitations on different BSs (infeasible problem).
In such cases, the UE with the minimum norm metric (w.r.t. all connected BSs) will be removed from its current NOMA cluster and join another cluster that uses a different sub-band (with different or the same set of cooperating BSs). 


\section{Numerical Results}
In this section, we provide some numerical results to discuss the performance of the proposed scheme under different system parameters, and then present illustrative results on the proposed NOMA transmission power allocation scheme.
Each value  is  obtained via  $2\times10^6$ Monte-Carlo simulation runs. For simplicity, 
we study only the case where all channel gains are i.i.d.  
Table \ref{Table3} presents the main network parameters used to obtain the simulation and analytical results.
\begin{table}[h!]
    \centering
      \caption{Simulation Parameters}
\begin{tabular}{|c|c|}
\hline
Parameter & Value  \\
\hline
\hline
AWGN PSD per UE &
$-169$ dBm/Hz  \\
\hline
Transmit power budget at a BS, $P$ & Variable 
\\
\hline
SIC sensitivity, $P_s$  & $1$ dBm
\\
\hline
SINR threshold per UE, $\gamma_{\text{th}}$  & $15$ dBm\\
\hline
Target SINR outage probability, $\epsilon$ & $10^{-5}$\\
\hline
\end{tabular}
    \label{Table3}
\end{table}
\subsection{Outage and Capacity Analysis}
In this section, we present  results on the outage and capacity performance of the proposed scheme.
We first start by evaluating the performance gain of the proposed GCoMP scheme compared to that of conventional CoMP system (considering both orthogonal multiple access [OMA] and NOMA paradigms).
Fig. \ref{Fig0} shows the average spectral efficiency per UE with different cooperation and multiple acces scenarios.
\begin{figure}[!htb]
		\centering
	\includegraphics[height=7.751cm, width=9.45cm]{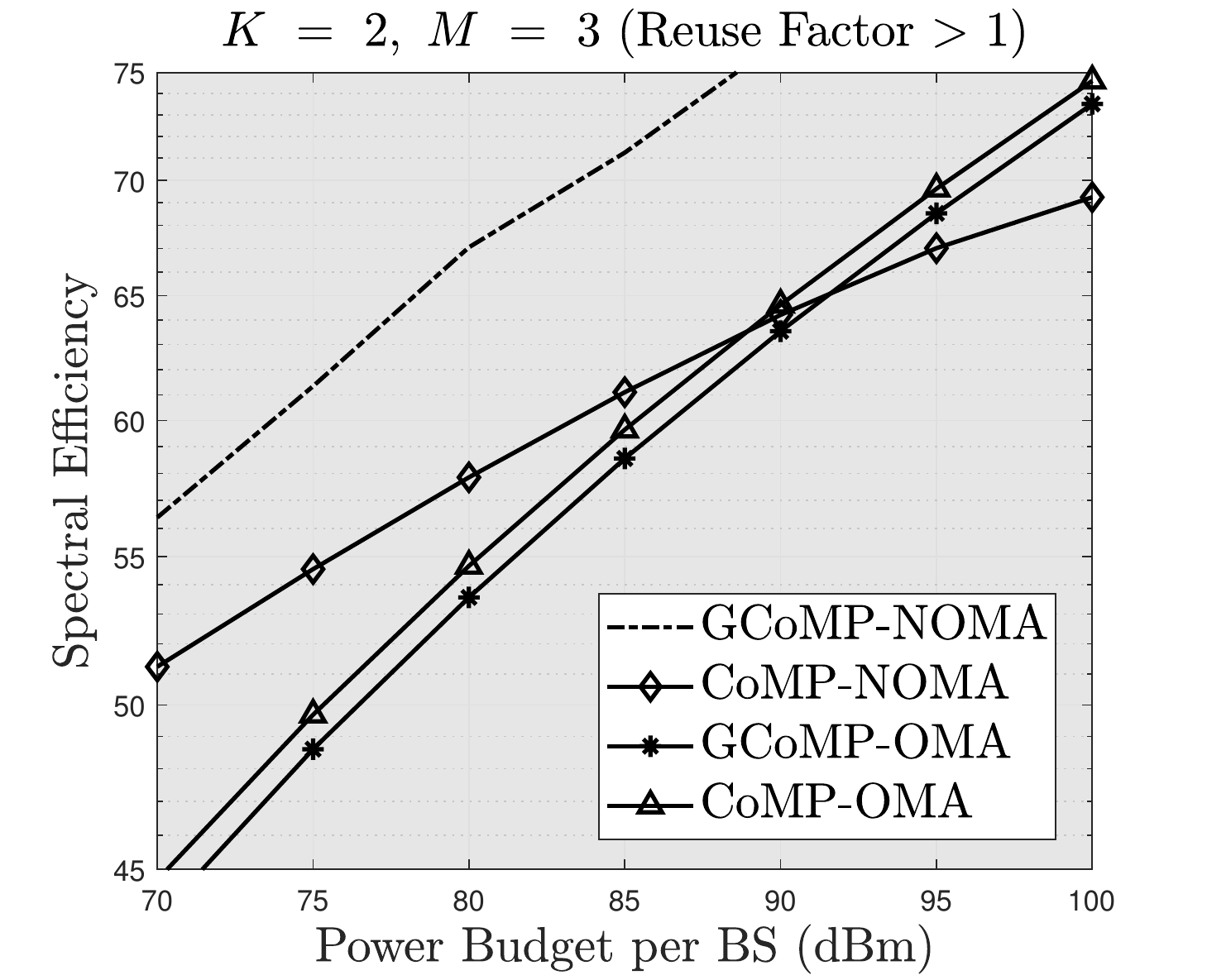}
		\caption{Average spectral efficiency with different cooperation and multiple acces scenarios.}\label{Fig0}
	\end{figure} 
It can be noticed from this figure that with OMA, a slight enhancement can be achieved when moving from CoMP into its generalized version (GCoMP) when $K=2$. 
However, more significant enhancement can be achieved for higher $K$.
Additionally, the GCoMP-NOMA system has been found to achieve the best performance for all power range while the performance of CoMP-NOMA scheme 
 lies between those of GCoMP-NOMA and GCoMP-OMA for relatively low power levels. The performance of CoMP-NOMA has been observed to degrade significantly as the maximum transmission power budget per BS increases due to interference caused by the non-cooperating BSs to the cell-centre UEs in other cells.
Note that to simulate CoMP-NOMA scheme and compare it with the other schemes in a fair manner, we have assumed that the interference power from the non-cooperating BSs to be $10^{-3}$ times the overall transmission power which takes into consideration the high distances between cell-centre UEs and other cells in a typical JT-CoMP scheme.  

To evaluate the $n{\text{-th}}$ order clustering scheme, Fig. \ref{Fig0} shows the probability of outage of the proposed system under different clustering levels versus the maximum transmission power budget per BS.
\begin{figure}[!htb]
		\centering
	\includegraphics[height=8.351cm, width=9.85cm]{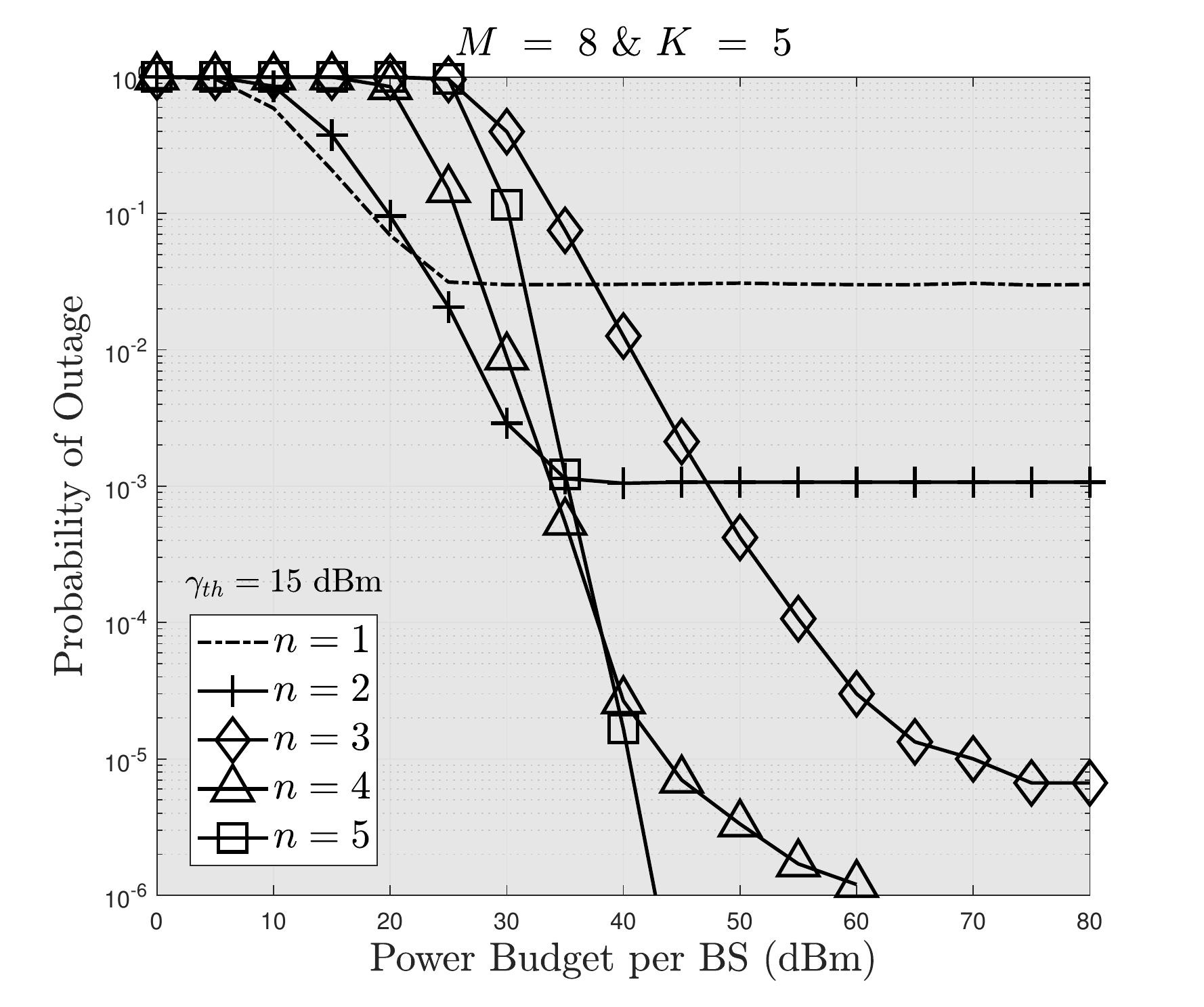}
		\caption{$\text{P}_{\text{out}}^{(n)}$ versus $P$ for different clustering order $n$.}\label{Fig0}
	\end{figure} 
It can be noticed from Fig. \ref{Fig0} that the outage performance is enhanced significantly when the clustering order increases. When many BSs, e.g.  $K=5$ are transmitting in the same cluster sub-band while the clustering order is low, the outage performance deteriorates. For example, a total service blockage occurs at $K=5$ and $n=1$. 
However, when the clustering order $n$ is relatively close to the number of cooperating BSs $(K)$, the interference on different UEs caused by non-serving BSs becomes tolerable. 

It is also important to investigate the effect of increasing the number of UEs per single NOMA cluster on the overall outage performance.
Fig. \ref{Fig1} shows the outage probability for a UE versus the maximum transmission power budget per BS under full-order clustering.
\begin{figure}[!htb]
		\centering
	\includegraphics[height=8.351cm, width=9.85cm]{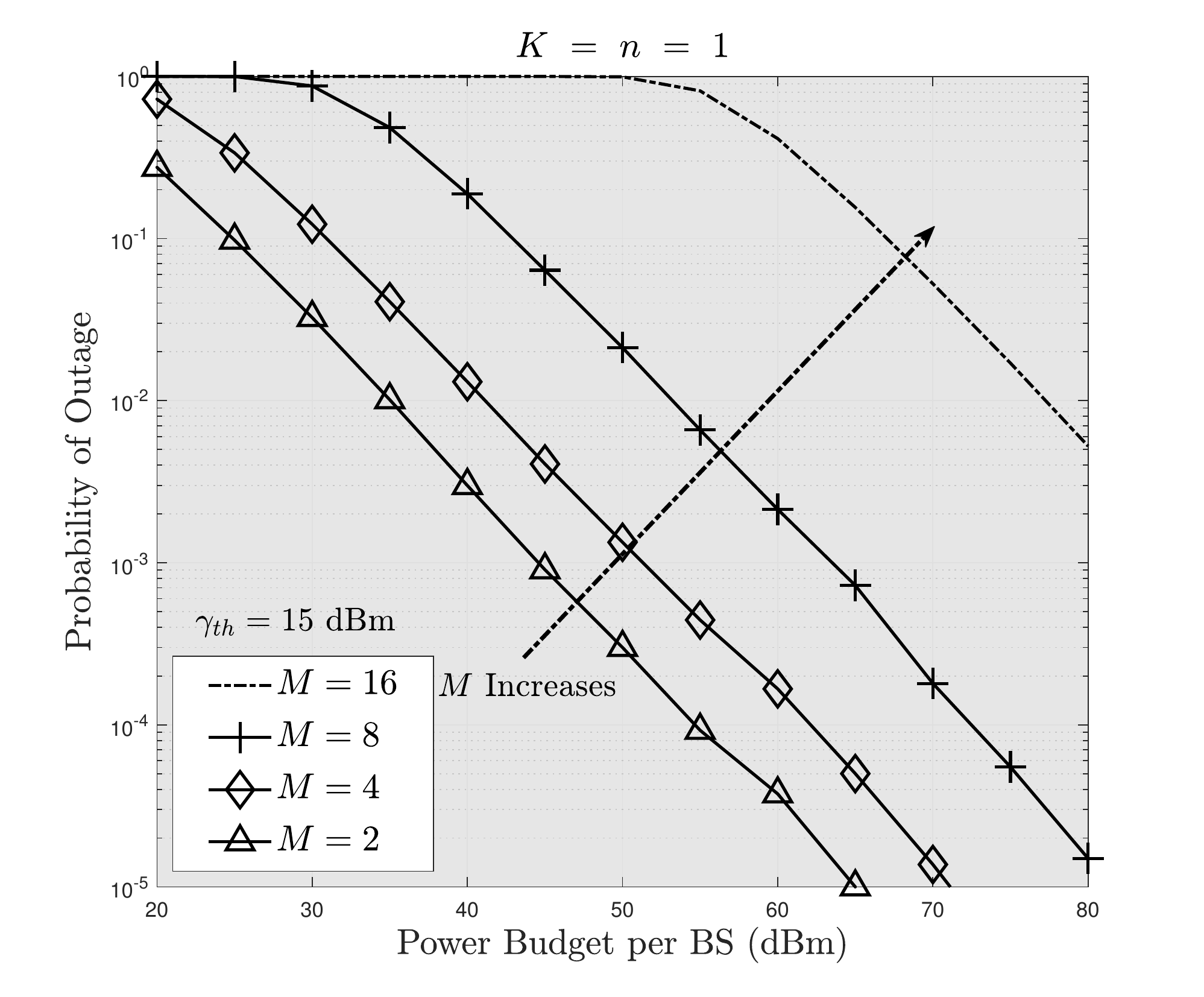}
		\caption{$\text{P}_{\text{out}}^{(n)}$ versus $P$ for different number of UEs $M$.}\label{Fig1}
	\end{figure} 
It can be noticed that when only one BS is serving a set of NOMA UEs, increasing the NOMA cluster size will cause a significant degradation in system coding gain (i.e. the outage performance curve shifts to the right).
This is because, as $M$ increases, the average number of interfering signal components with lower weights than the desired signal will increase.
This degradation prevents the potential use of NOMA access scheme in its conventional form (one serving BS) at the massive scale.

To compensate for the performance degradation caused by large NOMA cluster size, the number of serving BSs per cluster may be increased as in the proposed $n$-th order clustering scheme.
Fig. \ref{Fig2} shows the outage probability for a UE versus the maximum power budget per BS under different number of serving BSs and with a relatively large cluster size ($M=8$).
\begin{figure}[!htb]
		\centering
	\includegraphics[height=8.451cm, width=9.95cm]{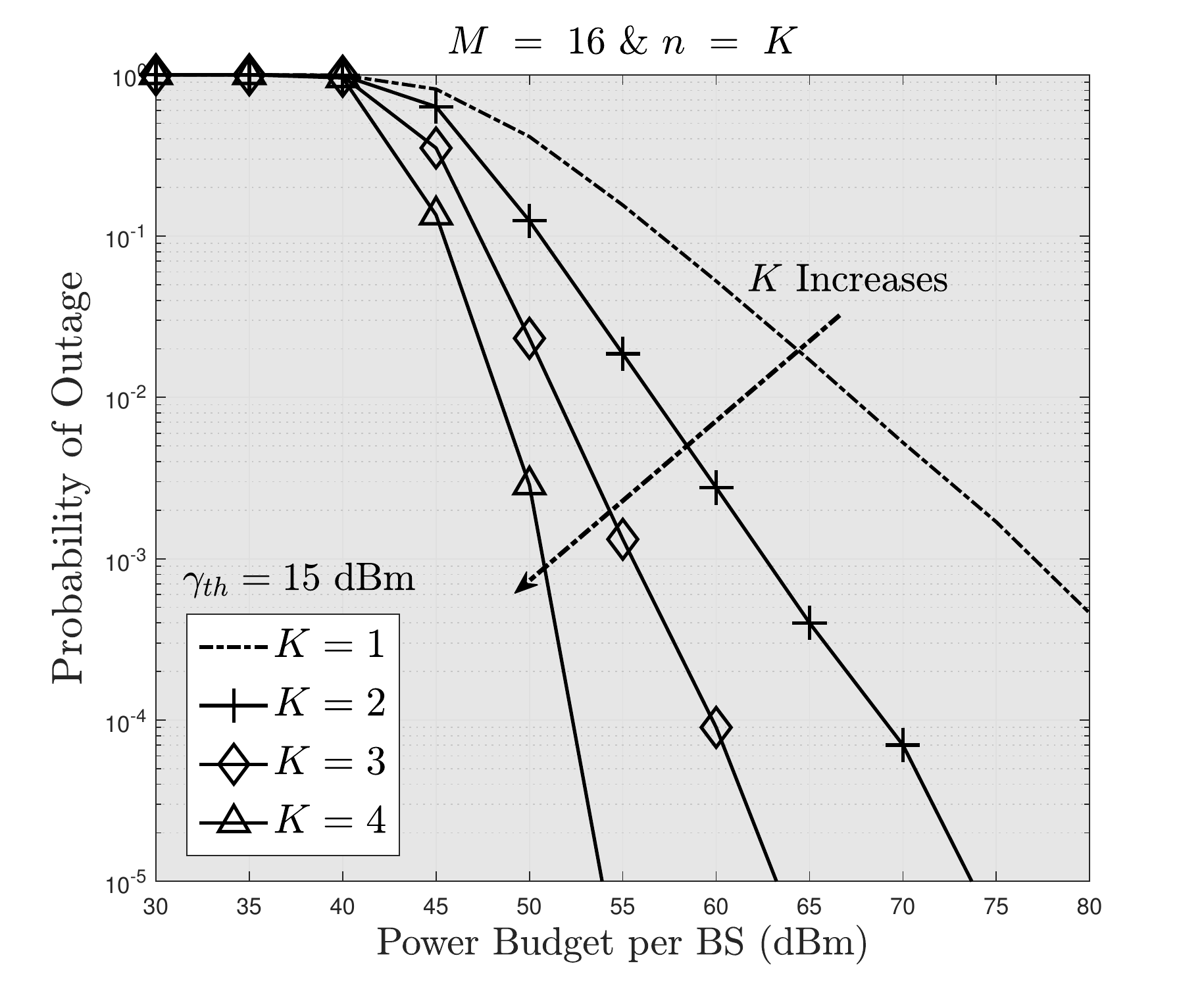}
		\caption{$\text{P}_{\text{out}}^{(n)}$ versus $P$ for different number of UEs $M$.}\label{Fig2}
	\end{figure} 
It can be noticed that a large number of NOMA clusters can co-exist in the same spectrum band when the number of cooperating BSs $(K)$ increases.

Finally, Fig. \ref{Fig5} shows the $\epsilon$-outage capacity ($C_{\epsilon}$) versus transmission power $P$ for different number of cooperating APs.
\begin{figure}[!htb]
		\centering		\includegraphics[height=8.451cm, width=9.95cm]{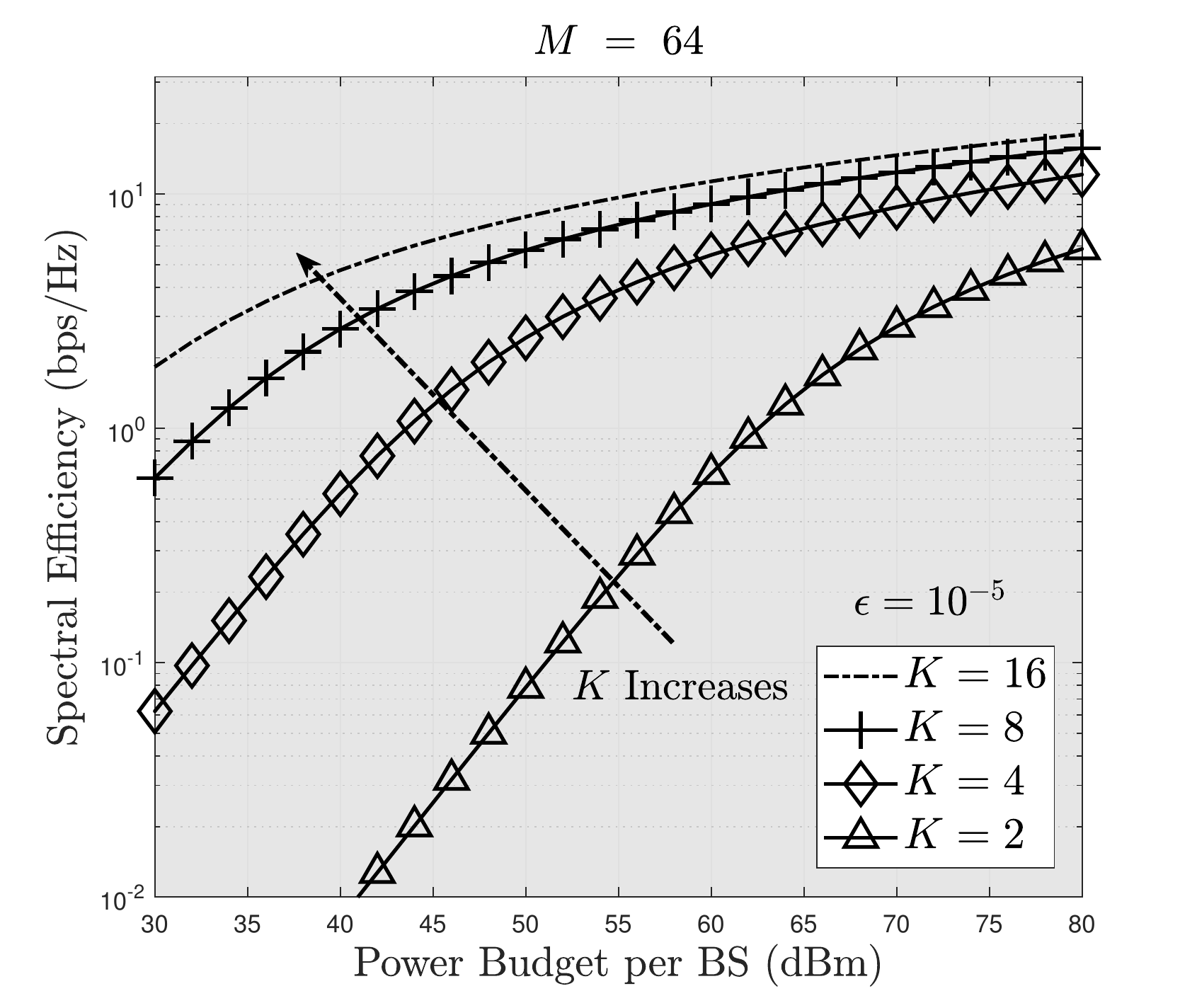}	\caption{$C_{\epsilon}$ versus $P$ for different number of BSs $K$.}\label{Fig5}
	\end{figure}
As was concluded from (\ref{Capacity}), increasing $K$ exponentially annihilates the negative effect caused by the large number of NOMA UEs using the same sub-band.
\subsection{Power Allocation}
In this section, we study the spectral efficiency of the simplified GCoMP-enabled NOMA  scheme with optimal transmission power allocation.
Additionally, the performance of the proposed GCoMP-enabled NOMA is compared with its OMA counterparts.
The optimal solution is obtained by solving the KKT conditions derived in Eqs. \ref{ConditionsFinal} in \textbf{Appendix D} for every channel realization and then selecting the set of feasible solutions (when transmission power budget is adequate to fulfill all the constraints of the optimization problem \textbf{J} in (\ref{OptimizationProb}).
The channel realization here is assumed to follow i.i.d fading distribution with unity variance.
For a given transmission power budget value, the overall spectral efficiency is calculated by averaging the accumulated sum-rate produced from realizations with feasible solutions for that power budget.
Additionally, the minimum power rate for every UE is considered as the achievable rate for the same UE using OMA with GCoMP. 

Fig. \ref{Fig_Opt_1} shows the optimized spectral efficiency (sum-rate) versus the maximum transmission power budget of GCoMP network layout with both multiple access schemes (NOMA and OMA).
\begin{figure}[!htb]
		\centering		\includegraphics[height=8.951cm, width=9.45cm]{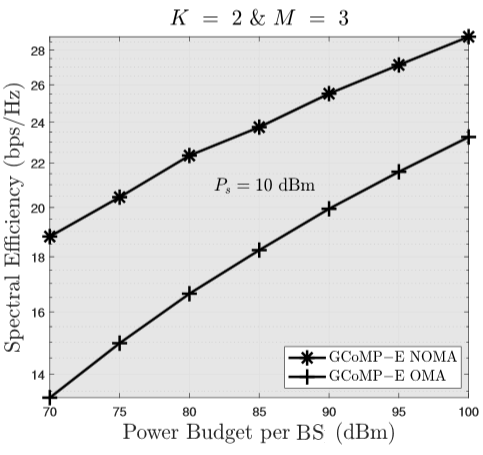}	\caption{Sum rate versus $P$ under optimal transmission power allocation.}\label{Fig_Opt_1}
	\end{figure}
As can be  noticed here, the sum-rate of the GCoMP-enabled NOMA scheme with $K=2$ and $M=3$ is superior to that of GCoMP with OMA scheme.
However, the spectral efficiency enhancement decreases as the maximum transmission power budget increases.
  
To study the performance of UEs individually, Fig. \ref{Fig_Opt_2} shows the spectral efficiency per UE under optimal transmission power allocation. Note that the UEs are ordered in an ascending mode from the least norm UE to the highest norm UE.
\begin{figure}[!htb]
		\centering	      	\includegraphics[height=8.751cm, width=9.45cm]{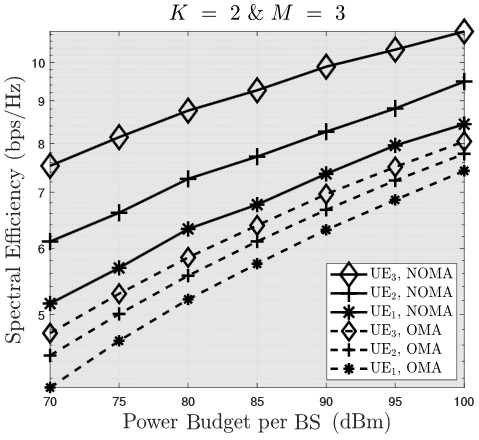}	\caption{UE rate versus $P$ under optimal transmission power allocation.}\label{Fig_Opt_2}
	\end{figure}
It can be noticed that a UE with a higher norm always yields a larger spectral efficiency compared to the one with a lower norm. 
This is because power allocation for NOMA  is based on inverse water-filling method until all UEs are provided with minimum rate requirements, and then water-filling is utilized. 
In particular, we notice that, the smaller the power budget per BS, the larger is the number of BSs which apply power allocation to the cluster members for NOMA, and the higher the power budget, the smaller is the number of BSs which apply power allocation to meet the requirements of NOMA (i.e. most of the BSs then use the ordinary water-filling power allocation).
Note that increasing the value of the desired SIC receiver sensitivity $P_s$ will decrease the number of feasible solutions for every channel realization, and hence, increases the processing time during real-time operation.  
    \section{Conclusion}
A novel generalized CoMP-enabled NOMA scheme has been proposed and evaluated.
In particular, the traditional joint-transmission CoMP scheme has been generalized to be applied for all users (i.e. cell-centre as well as cell-edge users) within the coverage area of a wireless network.
Furthermore, every base station has been assumed to apply a multi-user NOMA scheme for all users associated to it.
To evaluate the proposed scheme, the closed-form expressions for the probability of outage and outage capacity per user with different orders of BS cooperation have been derived.
To reduce the complexity of the proposed system, a low-complexity full order clustering protocol has been designed for the generalized CoMP-enabled NOMA system where the optimal transmission power allocation method has been obtained.
Findings show that it is possible to deploy NOMA with a large number of users per sub-band and tolerable complexity as long as the number of cooperating base stations is comparable to the number of NOMA users. Possible extensions of this work would include an evaluation of the  effect of imperfect CSI and synchronization on the system performance as well as a comprehensive evaluation and implementation  of the SIC methods for the proposed scheme. 

	\appendices
	\section{}
\setcounter{equation}{0}
	First, let us define $I_{\text{ICI}}^m=\sum_{w=n+1}^K\Phi_w|h_{m,w}|^2=\sum_{w=n+1}^Ky_{m,w}$, where $Y_{m,w}\thicksim \text{Exp}\left(\lambda_{m,w}\right)$ and $\lambda_{m,w}=1/2\Phi_w\sigma_{m,w}^2$. 
Using the theory of order statistics, the PDF of $I_{\text{ICI}}^m$ can be defined as
\begin{dmath}
    f_{I_{\text{ICI}}^m}(y_{m,n+1},\dots, y_{m,K})=\sum_{i_{n+1}, \dots, i_{K}}^{\{1,2, \dots, K\}}f_{Y_{m,i_{n+1}}}\left(y_{m,n+1}\right)\dots f_{Y_{m,i_{K}}}\left(y_{m,K}\right)\\
    \times 
    \prod_{j=1}^n\left(1-F_{Y_{m,i_j}}(y_{m,n+1})\right)\\
    =\sum_{i_{n+1}, \dots, i_{K}}^{\{1,2, \dots, K\}}\prod_{l=n+1}^{K}\lambda_{m,i_l}e^{-\lambda_{m,i_l}y_{m,l}}\prod_{j=1}^n e^{-\lambda_{m,i_j}y_{m,n+1}}
    ,\label{JointPDF1}
\end{dmath}
where $y_{m,K}\leq \dots \leq y_{m,n+1}$ and  $\{i_1, \dots, i_{K}\}$ are distinct indices that take values from $\{1, \dots, K\}$. 
Accordingly, the value of $\bar I_{\text{ICI}}$ can be defined as
\begin{dmath}
    \bar I^m_{\text{ICI}}=\sum_{i_{n+1}, \dots, i_{K}}^{\{1,2, \dots, K\}}
\int_0^{\infty}\int_{y_{m,K}}^{\infty}\dots \int_{y_{m,n+2}}^{\infty}
\sum_{i=n+1}^Ky_{m,i}
\\
\times
    \prod_{l=n+1}^{K}\lambda_{m,i_l}e^{-\lambda_{m,i_l}y_{l}}\prod_{j=1}^n\lambda_{m,i_l}e^{-\lambda_{m,i_j}y_{m,n+1}}d \dots dy_{m,K}.
\end{dmath}

Due to the dependence between $y_i's$, this integral cannot be changed into a product of independent integrals. 
Therefore, $\bar I_{\text{ICI}}^m$ can be rewritten as
\begin{dmath}
       \bar I^m_{\text{ICI}}=\sum_{i_{n+1}, \dots, i_{K}}^{\{1,2, \dots, K\}}
\int_0^{\infty}\dots \int_{0}^{\infty}\sum_{w=1}^{K-n}w x_w
\times
\left(
\prod_{j=n+1}^K\lambda_{i_j}e^{\left(\sum_{q=n+1}^j \lambda_{m,i_q}+\sum_{k=1}^n \lambda_{m,i_k} \right)x_j}
\right)\dots dx_{K-n},\label{IaveInt}
\end{dmath}
where we have used Sukhatme transformation of random variables (rv)s such that $x_i=y_{m,n+i}-y_{m,n+i+1}$ and $x_{n}=y_{m,K}$ represent an independent random variables \cite{CSUKHATME2012}. Due to the independence among $x$s, (\ref{IaveInt}) can be easily changed into a sum of a product of one-dimensional integral. Hence, we obtain (\ref{IFainal}). 
\section{}
\setcounter{equation}{0}
In (\ref{out1}), we have two types of ordering to be considered. The first one is the ordering of BSs w.r.t the $m{\text{-th}}$ UE. 
The second ordering is the ordering of the $m{\text{-th}}$ UE w.r.t all clusters it belongs to. 
First, we will consider the ordering of the BSs connected to the $m{\text{-th}}$ UE.
Accordingly, the PDF of $z_m$ can be defined as
\begin{dmath}
    f_{Z_m}(z_{m,1}, \dots, z_{m,n})=\sum_{i_{1}, \dots, i_{n}}^{\{1,2, \dots, K\}} f_{Z_{m,i_{1}}}\left(z_{m,1}\right)\dots f_{Z_{m,i_{n}}}\left(z_{m,n}\right) 
    \prod_{j=n+1}^K F_{Z_{m,i_j}}(y_{n})\\\
    =\sum_{i_{1}, \dots, i_{n}}^{\{1,2, \dots, K\}} \prod_{l=1}^{n}\alpha_{m,i_l}e^{-\alpha_{m,i_l}z_{m,l}}\prod_{j=n+1}^K\left(1- e^{-\alpha_{m,i_j}z_{m,n}}\right).
\end{dmath}
Utilizing the Sukhatme transformation, the MGF of $Z_m$ is given by
\begin{dmath}
    M_{Z_m}(s)=
    \sum_{i_{1}, \dots, i_{n}}^{\{1,2, \dots, K\}}
    \left[
    \prod_{q=1}^n \frac{\alpha_{m,i_q}}{\sum_{r=1}^q \alpha_{m,i_r}-q s}
    -
    \sum_{h_1=1}^{K-n}(-1)^{h_1}\sum^{\{n,\dots,K\}}_{j_1\leq \dots \leq j_{h_1}}
    \prod_{d=1}^{n}\frac{\alpha_{m,i_d}}{C_d^{m(h_1,i,j)}-d s} 
    \right],\label{MGF1}
\end{dmath}
where, $\{j_1, \dots, j_{h_1}\}$ are distinct ordered indices taking values from $\{n,\dots, K\}$ and $C_d^{m(h_1,i,j)}$ is defined as
\[
C^m_{d(h_1,i,j)}=\left\{ \begin{array}{cc}
 \sum_{r=1}^d \alpha_{m,i_r}    & d\in \left[1, \cdots, n-1\right] \\
  \sum_{r=1}^d \alpha_{m,i_r} + \sum_{l=1}^{h_1} \alpha_{m,j_l} & d=n
\end{array}
\right.
\].
To obtain the PDF $f_{Z_m}(z)$, it is convenient to express (\ref{MGF1}) as a partial fraction expression. Specifically, 
\begin{dmath}
    M_{Z_m}(s)=
    \sum_{i_{1}, \dots, i_{n}}^{\{1,2, \dots, K\}}
    \left( \prod_{q=1}^n \frac{\alpha_{m,i_q}}{q}\right)
    \left[
    \sum_{t_1=1}^n \frac{\eta_{t_1}^m}{\rho_{t_1}^m- s}
    -
    \sum_{h_1=1}^{K-n}(-1)^{h_1}\sum^{\{n,\dots,K\}}_{j_1\leq \dots \leq j_{h_1}}
    \sum_{t_2=1}^n
    \frac{\eta_{t_2}^m}{\rho_{t_2}^m - s} 
    \right],\label{MGF2}
\end{dmath}
where $\rho_{t_1}^m=\sum_{r=1}^{t_1}\alpha_{m,i_{r}}/t_1$, $\rho_{t_2}^m=C^m_{t_2}(h_1,i,j)/t_2$, $\eta_{t_1}^m=\prod_{k=1,k\neq t_1}^n\left(\rho_{k}^m-\rho_{t_1}^m \right)^{-1}$ and $\eta_{t_2}^m=\prod_{k=1,k\neq t_2}^n\left(\rho_{k}^m-\rho_{t_2}^m \right)^{-1}$. Note that (\ref{MGF2}) is valid only under the assumption that all $\rho_{t_1}^m$(and $\rho_{t_2}^m$) are distinct (the i.n.d case).
Upon finding the Laplace transform $L_{Z_m}(x)=M_{Z_m}(-x)$ and using the theory of inverse Laplace transform, the PDF of $Z_m$ (denoted by $f_{Z_m}(z)$) is then given by

\begin{dmath}
    f_{Z_m}(z)=
    \sum_{i_{1}, \dots, i_{n}}^{\{1,2, \dots, K\}}J_1(m,i)
    \left[
    \sum_{t_1=1}^n {\eta_{t_1}^m}e^{-\rho^m_{t_1}z}
    -
    \sum_{h_1=1}^{K-n}(-1)^{h_1}\sum^{\{n,\dots,K\}}_{j_1\leq \dots \leq j_{h_1}}
    \sum_{t_2=1}^n
    {\eta_{t_2}^m}e^{-\rho_{t_2}^mz} 
    \right],\label{PDF1}
\end{dmath}
\hspace{3mm}
where $J_1(m,i)=\left( \prod_{q=1}^n \frac{\alpha_{m,i_q}}{q}\right)$.
Now, we consider the ordering of the $m{\text{-th}}$ UE within the cluster of its best serving BS.
By utilizing the CDF expression of the $m{\text{-th}}$ order statistics for the set of i.n.d rvs given in \cite[Eq. 5.2.1]{OrderStcs2005}, the outage probability of the proposed system is given as in \textbf{Theorem 2}.
\section{}
\setcounter{equation}{0}
To proof the convexity of the problem in (\ref{OptimizationProb}), we need to first prove that the objective function is concave and then show that all constraints represents an affine transformation of dependent variables.
It is apparent that the objective function (let us denote it as $R_s$) is twice differentiable for all dependent variables ($a_{m,k}$).
However, due to the two-dimensional nature of dependent variables (functions of $m$ and $k$), finding the Hessian matrix of the partial second derivative even for the most simplified model ($M=2~\&~K=2$) would be very lengthy and tedious.
Therefore, we use an intuitive method to proof the concavity of $R_s$.

The objective function at (\ref{OptimizationProb}) can be rewritten as 
\begin{dmath}
    R_s=\sum_{m=1}^M \left[\underbrace{ \log_2\left(\theta_1(m)+1 \right)}_{f_1(m)}+\underbrace{\log_2\left(\frac{1}{\theta_2(m)+1} \right)}_{f_2(m)}
     \right],
\end{dmath}
where $\theta_1(m)=\sum_{k=1}^K\left(a_{m,k}+\sum_{j=m+1}^Ma_{j,k}\right)\gamma_{m,k}$ and $\theta_2(m)=\sum_{k=1}^K\left(\sum_{j=m+1}^M a_{j,k}\right)\gamma_{m,k}$.
It is apparent that $f_1(f_2)$ is a monotonically increasing (decreasing) function of $\theta_1(m)\left(\theta_2(m)\right)$ with $f_1$ being a strictly concave function and $f_2$ being a strictly convex function (logarithmic functions). 
Additionally, due to the set of constraints $\textbf{C}_2$, the overall value of $\theta_1(m)$ will be always strictly greater than that of $\theta_2(m)$ and any increase in $\theta_2(m)$ will result in a higher increase in $\theta_1(m)$.
Accordingly, the degree of convexity of $f_2(m)$ will be always greater that the degree of concavity of $f_1(m)$ which will make the summation $f_1(m)+f_2(m)$ to be always strictly convex (for all $a_{i,j}, i=1, \dots, M$ and $j=1, \dots, K$).
Finally, since $R_s$ represents a positive linear sum of convex functions, $R_s$ is a convex function as well.
Note that without the  constraints $\textbf{C}_2$, the problem in (\ref{OptimizationProb}) will be neither convex nor concave. 
Nevertheless, a global  optimal point for such non-convex problem can be found using the  `MAPEL' algorithm \cite{MAPEL}.

Now, we need to prove that the constraints $\textbf{C}_1$ through $\textbf{C}_3$ represent affine constraints.
It is apparent that the constraints $\textbf{C}_2$ and $\textbf{C}_3$ represent affine functions for all $a_{i,j}$. 
Additionally, the constraints $\textbf{C}_1$ can be rewritten as
$\sum_{k=1}^K a_{m,k}\gamma_{m,k}-\gamma_{\text{th}}^m\left(1+\sum_{k=1}^{K}\left(\sum_{j=m+1}^{M}a_{j,k}\right)\gamma_{m,k} \right)\geq 0$, where $\gamma_{\text{th}}^m=2^{R_m}-1$, which represent an affine function as well.
Hence, \textbf{Theorem \ref{Theorem2}} is proved.

\section{}
\setcounter{equation}{0}
Since the problem \textbf{J} in (\ref{OptimizationProb}) is convex with affine constraints, then the Karush-Kuhn-Tucker (KKT) conditions can be given by taking the partial derivative of (\ref{Lagrangian}) w.r.t. $\{a_{m,k}\}_{m=1,2,3}^{k=1,2}$, $\{\eta_m\}_{m=1,2,3}$, $\mu_1$, $\{\mu_i\}_{i=1,2}$ and $\{\tau_k\}_{k=1,2}$ as follows:
\begin{dgroup}\label{ConditionsFinal}
\begin{dmath}
\frac{\partial \mathcal{L}}{\partial a_{1,1}}=\frac{\gamma_{1,1}}{\sum_{k=1}^2\left(\sum_{j=1}^3 a^*_{j,k} \right)\gamma_{1,k}+1}+\eta^*_1 \gamma_{1,1}-\mu^*_1\gamma_{2,1}-\mu^*_2\gamma_{3,1}-\tau^*_1
\begin{split}
    \leq 0
\end{split},
\end{dmath}
\begin{dmath}
\frac{\partial \mathcal{L}}{\partial a_{1,2}}=\frac{\gamma_{1,2}}{\sum_{k=1}^2\left(\sum_{j=1}^3 a^*_{j,k} \right)\gamma_{1,k}+1}+\eta^*_1 \gamma_{1,2}-\mu^*_1\gamma_{2,2}-\mu^*_2\gamma_{3,2}-\tau^*_2
\begin{split}
\leq 0    
\end{split},
\end{dmath}
\begin{dmath}
\frac{\partial \mathcal{L}}{\partial a_{2,1}}=\frac{-\gamma_{1,1}\left(\sum_{k=1}^2a^*_{1,k}\gamma_{1,k} \right)}{\left(\sum_{k=1}^2\left[\sum_{i=1}^3a^*_{i,k}\right]\gamma_{1,k}+1\right) \left(\sum_{k=1}^2\left[\sum_{j=2}^3a^*_{j,k}\gamma_{1,k}\right]+1 \right)}
+
\frac{\gamma_{2,1}}{\sum_{k=1}^2\left(\sum_{j=2}^3 a^*_{j,k} \right)\gamma_{2,k}+1}
-\eta^*_1\gamma_{\text{th}}^1\gamma_{1,1} +\eta^*_2 \gamma_{2,1}+\mu^*_1\gamma_{2,1}+\mu^*_2\gamma_{3,1}-\mu^*_3\gamma_{3,1}-\tau^*_1
\begin{split}
    \leq 0
\end{split},
\end{dmath}
\begin{dmath}
\frac{\partial \mathcal{L}}{\partial a_{2,2}}=\frac{-\gamma_{1,2}\left(\sum_{k=1}^2a^*_{1,k}\gamma_{1,k} \right)}{\left(\sum_{k=1}^2\left[\sum_{i=1}^3a^*_{i,k}\right]\gamma_{1,k}+1\right) \left(\sum_{k=1}^2\left[\sum_{j=2}^3a^*_{j,k}\gamma_{1,k}\right]+1 \right)}
+
\frac{\gamma_{2,2}}{\sum_{k=1}^2\left(\sum_{j=2}^3 a^*_{j,k} \right)\gamma_{2,k}+1}
-\eta^*_1\gamma_{\text{th}}^1\gamma_{1,2} +\eta^*_2 \gamma_{2,2}+\mu^*_1\gamma_{2,2}+\mu^*_2\gamma_{3,2}-\mu^*_3\gamma_{3,2}-\tau^*_2
\begin{split}
    \leq 0
\end{split},
\end{dmath}
\begin{dmath}
\frac{\partial \mathcal{L}}{\partial a_{3,1}}=\frac{-\gamma_{1,1}\left(\sum_{k=1}^2a^*_{1,k}\gamma_{1,k} \right)}{\left(\sum_{k=1}^2\left[\sum_{i=1}^3a^*_{i,k}\right]\gamma_{1,k}+1\right) \left(\sum_{k=1}^2\left[\sum_{j=2}^3a^*_{j,k}\gamma_{1,k}\right]+1 \right)}
+
\frac{-\gamma_{2,1}\left(\sum_{k=1}^2a^*_{2,k}\gamma_{2,k} \right)}{\left(\sum_{k=2}^2\left[\sum_{j=2}^3a^*_{j,k}\right]\gamma_{2,k}+1\right)\left(\sum_{k=1}^2a^*_{3,k}\gamma_{2,k}+1\right)}
+
\frac{\gamma_{3,1}}{\sum_{k=1}^2a^*_{3,k}\gamma_{3,k}+1}
-
\eta^*_2 \gamma_{2,1}+\eta^*_3\gamma_{3,1}-\mu^*_3\gamma_{3,1}-\tau^*_1
\begin{split}
\leq 0    
\end{split}
,
\end{dmath} 
\begin{dmath}
    \frac{\partial \mathcal{L}}{\partial a_{3,2}}= \frac{-\gamma_{1,2}\left(\sum_{k=1}^2a^*_{1,k}\gamma_{1,k} \right)}{\left(\sum_{k=1}^2\left[\sum_{i=1}^3a^*_{i,k}\right]\gamma_{1,k}+1\right) \left(\sum_{k=1}^2\left[\sum_{j=2}^3a^*_{j,k}\gamma_{1,k}\right]+1 \right)}
+
\frac{-\gamma_{2,2}\left(\sum_{k=1}^2a^*_{2,k}\gamma_{2,k} \right)}{\left(\sum_{k=2}^2\left[\sum_{j=2}^3a^*_{j,k}\right]\gamma_{2,k}+1\right)\left(\sum_{k=1}^2a^*_{3,k}\gamma_{2,k}+1\right)}
+
\frac{\gamma_{3,2}}{\sum_{k=1}^2a^*_{3,k}\gamma_{3,k}+1}
-
\eta^*_2 \gamma_{2,2}+\eta^*_3\gamma_{3,2}-\mu^*_3\gamma_{3,2}-\tau^*_2
\begin{split}
    \leq 0
\end{split},
\end{dmath}
\begin{dmath}
\begin{split}
\frac{\partial \mathcal{L}}{\partial \eta_q}=
\sum_{k=1}^2\left(a^*_{q,k}-\gamma_{\text{th}}^q\sum_{j=q+1}^3a^*_{j,k} \right)\gamma_{q,k}-\gamma_{\text{th}}^q    
\end{split}
\geq 0,
\end{dmath}
\begin{dmath}
\begin{split}
\frac{\partial \mathcal{L}}{\partial \mu_1}=P_s-\sum_{k=1}^2\left(a^*_{1,k}-a^*_{2,k} \right)\gamma_{2,k}    
\end{split}
\geq 0,
\end{dmath}
\begin{dmath}
\begin{split}
    \frac{\partial \mathcal{L}}{\partial \mu_d}=P_s-\sum_{k=1}^2\left(a^*_{d-1,k}-\sum_{l=d}^3a^*_{l,k} \right)\gamma_{3,k}
\end{split}\geq 0,
\end{dmath}
\begin{dmath}
\begin{split}
\frac{\partial \mathcal{L}}{\partial \tau_v}=1-\sum_{m=1}^3a^*_{m,v}    
\end{split}  \geq 0,
\end{dmath}
\end{dgroup}
where $q=1, 2, 3$, $d=2, 3$ and $v=1, 2$.   
The set of points $\{a_{m,k}^*\}_{m=1,2,3}^{k=1,2}\geq 0$, $\{\eta_m^*\}_{m=1,2,3}\geq 0$, $\mu_1^*\geq 0$, $\{\mu_i^*\}_{i=1,2}\geq 0$ and $\{\tau_k^*\}_{k=1,2}\geq 0$ that satisfy conditions (\ref{ConditionsFinal}) are both the primal and dual optimal solutions for problem $\mathbf{J}$ in (\ref{OptimizationProb}).
Since the set of equations (\ref{ConditionsFinal}) are differentiable, we may utilize one of the numerical methods used for solving a set of differentiable non-linear equations such as Newton or Broyden methods.

\bibliographystyle{IEEEtran}
\bibliography{IEEEabrv,yasser}

\end{document}